\documentclass[journal]{IEEEtran}

\usepackage{graphicx}
\usepackage{amsmath,amssymb,amsthm}
\usepackage{enumerate}
\usepackage{enumitem}
\usepackage{booktabs}
\usepackage{algorithm}
\usepackage{algpseudocode}
\usepackage{cite}

\newtheorem{theorem}{Theorem}
\newtheorem{lemma}{Lemma}

\newtheorem{remark}{Remark}
\newtheorem{assumption}{Assumption}

\newcommand{\R}{\mathbb{R}}
\newcommand{\E}{\mathbb{E}}
\newcommand{\Proj}{\Pi}
\newcommand{\norm}[1]{\left\|#1\right\|}
\newcommand{\abs}[1]{\left|#1\right|}
\newcommand{\cQ}{\mathcal{Q}}
\newcommand{\cS}{\mathcal{S}}
\newcommand{\cU}{\mathcal{U}}
\newcommand{\cX}{\mathcal{X}}
\newcommand{\cN}{\mathcal{N}}
\newcommand{\cF}{\mathcal{F}}
\newcommand{\cT}{\mathcal{T}}

\begin{document}

\title{A Stateful Stochastic Allocation Mechanism with Fairness Guarantees\\
for Networked Electricity Systems}

\author{Shaun~Sweeney%
\thanks{S.~Sweeney is with the Dyson School of Design Engineering,
Imperial College London, UK.
Email: fight@enleashed.tech.
A companion paper, published in \textit{Energy Economics}~\cite{sweeney2026energyecon},
develops the economic interpretation of the mechanism, including
cost-causation pricing, distributional outcomes, revenue adequacy,
and comparisons against locational marginal pricing on a stylised
GB transmission network. The present paper focuses exclusively on
the control-theoretic properties of the mechanism: stability,
contraction, event-triggered execution, and dynamic fairness
convergence. None of Theorems~1--5 of the present paper appear
in the companion paper; the two papers are mathematically
self-contained with respect to each other.}}

\maketitle

% ------------------------------------------------------------------
\begin{abstract}
This paper develops and analyses the \emph{Fair Play Automatic Market
Maker} (FP-AMM), a programmable electricity allocation mechanism in
which scarcity allocation is treated as a controlled, stateful, and
auditable cyber-physical process. The problem addressed is that
dominant mechanisms such as locational marginal pricing are memoryless:
they carry no representation of historical service outcomes and therefore
cannot guarantee equitable treatment across market intervals. The
FP-AMM operates through a two-stage stochastic clearing rule---service-level
priority sampling followed by inverse-fairness weighting within each
tier---applied against a DC-OPF feasibility set, with bounded shortage
memory updated at each interval by a saturated integrator. Four main
results are established. First, the shortage-memory state is invariant
in $[0,1]^N$ and the update map is a contraction with explicit rate
$1-\beta$ (Theorem~\ref{thm:memory_bounded}). Second, the
intra-interval clearing operator converges linearly to a unique fixed
point with explicit factor $q \in (0,1)$
(Theorem~\ref{thm:operator_contraction}). Third, under the Fair Play
priority rule the per-node delivery ratio converges almost surely to
the contracted target $F^\star$, with a finite-time $O(1/\sqrt{T})$
bound and explicit constant established via a Lyapunov analysis of the
deficit recursion (Theorem~\ref{thm:fairness_convergence}). Fourth, event-triggered execution yields practical
ultimate boundedness of the allocation tracking error with an explicit
computation--fidelity trade-off (Theorem~\ref{thm:event_trigger}). The
mechanism is validated on IEEE 14-, 57-, and 118-bus networks over
$T = 5000$ market intervals: fairness convergence to $F^\star$ is
achieved on all benchmarks, Fair Play reduces peak weak-bus fairness
error by 54\% on the IEEE-57 network and up to 55\% relative to an
equal-weight baseline during scarcity windows, and DC feasibility is
maintained throughout.
\end{abstract}


% ======================================================================
\section{Introduction}
\label{sec:intro}
% ======================================================================

Electricity markets operate as cyber-physical systems in which
allocation decisions must simultaneously respect hard network
constraints, respond to volatile renewable supply and demand, and
deliver differentiated service levels under scarcity. The dominant
mechanism---locational marginal pricing (LMP)---collapses feasibility,
allocation, and settlement into nodal dual variables. Under historical
conditions of largely dispatchable supply and weakly active network
constraints, LMP provided a tractable approximation. Under high
renewable penetration, distributed energy resources, and persistently
active network constraints, this approximation breaks down: LMP signals
may be discontinuous, unbounded under scarcity, and memoryless with
respect to service history, offering no formal guarantee of equitable
treatment across market intervals~\cite{schweppe1988spot,stof2002pse,
wood2013power}.

This paper proposes a different paradigm. Market clearing is treated
as a stateful feedback process: the mechanism carries bounded memory
of historical service outcomes at each network node and uses it to
generate probabilistic allocation priorities that are explicitly fair,
physically feasible, and formally analysable as a networked control
system. The core mechanism is the \emph{Fair Play Automatic Market
Maker} (FP-AMM), in which a two-stage stochastic allocation rule
governs how residual flexible capacity is distributed across nodes
under scarcity.

\subsection{Physical and Control-Theoretic Motivation}

The central difficulty in modern electricity market design is the
decoupling of economic equilibrium from physical equilibrium. In
control-theoretic terms, the market historically supplied the outer
optimisation layer while primary and secondary frequency controls
regulated the inner physical dynamics. Variable renewables, distributed
energy resources, and active network constraints break this separation:
economic equilibrium of the market-clearing problem need not coincide
with physical equilibrium of the grid.

The FP-AMM addresses this by treating allocation as a cyber-physical
control loop with three composable layers: (i) a scarcity-responsive
AMM price signal that is bounded and Lipschitz by design; (ii)
service-level priority queues that implement contractual QoS guarantees
probabilistically, without hard capacity partitions; and (iii) a
stateful fairness memory that drives convergence of long-run delivery
ratios to contracted targets. This paper focuses exclusively on the
control-theoretic properties of the mechanism: convergence, stability,
and formal fairness guarantees. Economic evaluation---including
cost-causation analysis, distributional outcomes, bill dispersion,
generator revenue concentration, and comparisons with locational
marginal pricing on a stylised GB transmission network---is developed
in a companion paper published in \textit{Energy
Economics}~\cite{sweeney2026energyecon}. The present paper takes the
market mechanism as given and studies its properties as a networked
control system; none of the control-theoretic results (Theorems~1--5)
appear in the companion paper.

\subsection{Comparison with Prior Art}

\subsubsection{LMP and Primal-Dual Methods}
LMP arises as the dual variable of the nodal power-balance constraint
in OPF~\cite{schweppe1988spot,stof2002pse}. Standard primal--dual
methods produce price signals as dual variables; under scarcity these
may be discontinuous and lack convergence guarantees under time-varying
constraints~\cite{wood2013power,bergen2000power}. The FP-AMM differs
in three structural respects: its scarcity signal $p_t = f(\Delta_t)$
is saturated and Lipschitz by construction; its clearing operator
carries bounded internal state encoding service history; and fairness
is proved to converge in finite time rather than being asserted ex post.

\subsubsection{Online Resource Allocation and Network Utility
Maximisation}
Networked resource allocation has been extensively studied via
primal-dual and flow-control
frameworks~\cite{low1999flowcontrol,palomar2006tutorial}. These
typically establish convergence of an aggregate welfare objective under
stationarity assumptions, without explicit QoS tier differentiation or
historical fairness correction. The stochastic two-stage sampling rule
in the FP-AMM is closer to weighted fair queuing in network
scheduling~\cite{parekh1993generalized} but operates under physical
feasibility constraints inherited from the power system.

\subsubsection{Event-Triggered and Hybrid Control}
Event-triggered allocation and asynchronous market clearing have
attracted recent attention~\cite{tabuada2007eventtriggered,
heemels2013periodicetc}. The FP-AMM uses an event-driven update
schedule that reduces computation while maintaining bounded
service-history tracking. Theorem~\ref{thm:event_trigger} provides
explicit practical ultimate boundedness bounds under this schedule.

\subsection{Contributions}

This paper makes four contributions:
\begin{enumerate}[label=(\roman*)]
  \item \textbf{A stateful stochastic allocation mechanism}
        (Section~\ref{sec:mechanism}): a complete specification of the
        FP-AMM two-stage clearing rule for constrained networked
        systems, with bounded shortage memory and differentiated
        quality-of-service tiers. The mechanism is self-contained and
        does not defer its mathematical definition to a companion paper.
  \item \textbf{Contraction and practical stability guarantees}
        (Theorems~\ref{thm:memory_bounded}--\ref{thm:event_trigger}):
        the shortage-memory state is invariant and contractive; the
        intra-interval clearing operator converges linearly to a unique
        fixed point with explicit rate $q^k$; and event-triggered
        execution yields practical ultimate boundedness with an explicit
        computation--fidelity trade-off.
  \item \textbf{Almost-sure fairness convergence}
        (Theorem~\ref{thm:fairness_convergence}): under the Fair Play
        stochastic priority rule, the per-node long-run delivery ratio
        converges almost surely to $F^\star$ with finite-time
        $O(1/\sqrt{T})$ performance guarantees and explicit constant
        $C_\rho$, established via a Lyapunov analysis of the
        deficit recursion.
  \item \textbf{IEEE benchmark validation}
        (Section~\ref{sec:ieee}): the FP-AMM is validated on
        IEEE 14-, 57-, and 118-bus networks under DC-OPF feasibility
        constraints over $T = 5000$ market intervals. Fairness
        convergence to $F^\star$ is demonstrated on all benchmarks,
        with Fair Play ON reducing peak weak-bus error by up to
        55\% relative to an equal-weight baseline during scarcity
        windows. DC feasibility is maintained throughout; the
        event-triggered operator skips at least 25\% of intervals
        on every benchmark.
\end{enumerate}

The novelty of the FP-AMM does not lie in projected-gradient methods
themselves, which are standard, but in embedding a stateful fairness
recursion inside a feasibility-constrained stochastic allocation
mechanism and proving convergence of cumulative delivery ratios.
\section{Background and Related Work}
\label{sec:background}
% ======================================================================

\subsection{Network Feasibility and Market Clearing}

Operational market clearing is coupled to feasibility models
(OPF/SCOPF) that enforce power balance and network
constraints~\cite{wood2013power,bergen2000power}. Service
differentiation under scarcity is handled by additional rules outside
the price signal, without formal convergence guarantees on fairness
dynamics. The FP-AMM integrates these into a single mechanism with
formal guarantees.

\subsection{Fairness in Networked Allocation}

Fairness in networked resource allocation has been studied primarily
through max-min and proportional fairness
objectives~\cite{bertsekas1992data,kelly1998rate}, typically defined
with respect to a single-snapshot optimisation. The FP-AMM adopts a
dynamic fairness criterion: the target is that cumulative delivery
ratios converge to contracted shares over time, not merely that
instantaneous shares are proportional.

\subsection{Weighted Fair Queuing}

The two-stage stochastic sampling rule is structurally related to
weighted fair queuing (WFQ) and deficit round-robin scheduling in
packet networks~\cite{parekh1993generalized,shreedhar1996efficient}.
The FP-AMM extends WFQ to a physical power network with time-varying
feasibility constraints and coupled multi-interval shortage memory.
The key distinction is that QoS weights are dynamic: they incorporate
real-time historical fairness deficits rather than static contract
weights alone.

\subsection{Stochastic Approximation}

The fairness convergence result uses stochastic approximation
tools~\cite{robbins1951,borkar2008stochastic}. The shortage-memory
update plays the role of an averaging recursion; the convergence rate
follows from the contraction properties of the update map and a
martingale argument on the service-delivery sequence.

% ======================================================================
\section{Preliminaries and Notation}
\label{sec:prelims}
% ======================================================================

Let $\cN = \{1,\dots,N\}$ index nodes of the power network. Let
$t \in \{0,1,2,\dots\}$ index discrete market intervals
(e.g.\ 30~min). Let $\cQ_t$ denote the set of \emph{active flexible
requests}: requests whose permissible time windows intersect
$[t_0,t_1]$ and that have not yet been scheduled.

\subsection{Request Attributes}

Each request $i \in \cQ_t$ has: node $n(i) \in \cN$; service tier
$p(i) \in \{1,\dots,C\}$ (ordered premium-to-basic); power demand
$P_i > 0$ and duration $\Delta_i$ (in market slots); permissible
start window $[t_i^{\min}, t_i^{\max}]$; and historical fairness
ratio $F_{n(i)}(t) \in [0,1]$ (defined below).

\subsection{Scarcity Signal}

Let $d_t \in \R^N_{\ge 0}$ and $s_t^{\mathrm{avail}} \in \R^N_{\ge 0}$
denote aggregate requested demand and deliverable supply. The deficit
signal is
\begin{equation}
\Delta_t := d_t - s_t^{\mathrm{avail}},
\label{eq:deficit}
\end{equation}
with $\Delta_{t,n} > 0$ indicating shortage at node $n$. The AMM
broadcasts
\begin{equation}
p_t = f(\Delta_t), \qquad f(0) = 0,
\label{eq:price_law}
\end{equation}
where $f : \R^N \to \R^N$ is monotone, bounded ($\norm{f(\Delta)} \le
p_{\max}$), and $L_f$-Lipschitz. The canonical instantiation
$f(\Delta) = p_{\max}\tanh(\Delta/\sigma)$ has $L_f = p_{\max}/\sigma$.

\subsection{Feasibility Set}

The network feasibility set at interval $t$ is $\cU(x_t)$, consistent
with a DC-OPF relaxation of the power-flow
equations~\cite{wood2013power,bergen2000power}. After allocating
essential (non-curtailable) load, the residual flexible capacity at
node $n$ over $[t_0,t_1]$ is $S_{t,n}$.

\subsection{Historical Fairness Ratio and Shortage Memory}

Cumulative delivered and desired energy at node $n$ up to interval $t$
are $E_n^{\mathrm{del}}(t)$ and $E_n^{\mathrm{des}}(t)$. The
\emph{historical fairness ratio} is
\begin{equation}
F_n(t) := \frac{E_n^{\mathrm{del}}(t)}{E_n^{\mathrm{des}}(t)}
\in [0,1],
\label{eq:fairness_ratio}
\end{equation}
with $F_n(t) = 0$ when $E_n^{\mathrm{des}}(t) = 0$; note that under
Assumption~\ref{ass:participation} below, $E_n^{\mathrm{des}}(t) \ge
E_{\min} > 0$ for all $t \ge 1$, so this initialisation convention
applies only at $t = 0$. The target is $F^\star = 1$. The
\emph{fairness deficit} at node $n$ is
\begin{equation}
z_{n,t} := \max\!\bigl(0,\, F^\star - F_n(t)\bigr) \in [0,1],
\label{eq:fairness_deficit}
\end{equation}
and $z_t = (z_{n,t})_{n \in \cN} \in [0,1]^N$ is the
\emph{shortage-memory state}. The full system state is
$x_t = (x_t^{\mathrm{phys}}, z_t)$, where $x_t^{\mathrm{phys}}$
aggregates supply availability, demand, and network data.

% ======================================================================
\section{The Fair Play Automatic Market Maker}
\label{sec:mechanism}
% ======================================================================

The FP-AMM operates at each market interval $t$ as a two-stage
stochastic clearing rule on $\cQ_t$, followed by a shortage-memory
update.

\subsection{Stage 1: Service-Level Priority Sampling}

Service tiers are assigned non-negative priority weights
$w(1) \ge w(2) \ge \cdots \ge w(C) \ge 0$. The sub-queue for tier
$c$ is $\cQ_t^{(c)} := \{i \in \cQ_t : p(i) = c\}$. The mechanism
first selects a tier:
\begin{equation}
\Pr(\text{select tier } c) = \frac{w(c)}{\sum_{c'=1}^C w(c')}.
\label{eq:tier_sampling}
\end{equation}
Over repeated clearing events, tier-$c$ requests are attempted in
proportion to $w(c)$, without hard capacity partitioning.

\subsection{Stage 2: Inverse-Fairness Weighting Within Tier}

Given tier $c$, each request $i \in \cQ_t^{(c)}$ is assigned a
\emph{Fair Play priority score}:
\begin{equation}
S_{i,t} := w(p(i)) \cdot \bigl(\varepsilon_{\mathrm{b}} + z_{n(i),t}\bigr)^{\alpha_f},
\label{eq:priority_score}
\end{equation}
where $\varepsilon_{\mathrm{b}} > 0$ prevents starvation of well-served nodes and
$\alpha_f \ge 1$ controls sensitivity to historical deficits. Scores
are normalised:
\begin{equation}
\Pr(i \mid c) := \frac{S_{i,t}}{\sum_{j \in \cQ_t^{(c)}} S_{j,t}},
\qquad i \in \cQ_t^{(c)}.
\label{eq:selection_prob}
\end{equation}
Higher service tiers (larger $w(p(i))$) and more historically
under-served nodes (larger $z_{n(i),t}$) receive higher selection
probability.

\subsection{Feasibility Check and Scheduling}

A selected request $i$ is attempted by solving a local feasibility
problem: find $\tau_i \in [t_i^{\min}, t_i^{\max}]$ such that power
$P_i$ for duration $\Delta_i$ respects $S_{t,n(i)}$ and network
constraints. The cheapest-start criterion minimises the cost signal:
\begin{equation}
\tau_i^\star \in \arg\min_{\tau \in \cS_i}
  \sum_{s=\tau}^{\tau+\Delta_i-1} c[s],
\quad S_{t,n(i)}[s] \ge P_i \;\forall s.
\label{eq:cheapest_start}
\end{equation}
If no feasible start exists, $i$ is marked infeasible. Upon
acceptance, residual capacity is updated:
\begin{equation}
S_{t,n(i)}[s] \leftarrow S_{t,n(i)}[s] - P_i,
\quad s \in [\tau_i^\star,\, \tau_i^\star + \Delta_i - 1].
\label{eq:capacity_update}
\end{equation}

\subsection{Shortage-Memory Update}

After each interval, the fairness deficit is updated. Writing
$\ell_{n,t} := E_{n,t}^{\mathrm{des}} - E_{n,t}^{\mathrm{del}} \ge 0$
for the per-interval shortage, the running-average update is:
\begin{equation}
z_{n,t+1} = \Proj_{[0,1]}\!\bigl((1-\beta_t)\,z_{n,t}
  + \beta_t\,\ell_{n,t}\bigr),
\label{eq:memory_update}
\end{equation}
where $\beta_t = 1/(t+1)$ for a running average or $\beta \in (0,1)$
for exponential forgetting.

\subsection{Algorithm and Closed-Loop Representation}

Algorithm~\ref{alg:fair_play} summarises the complete FP-AMM clearing
procedure. The resulting closed-loop system is:
\begin{equation}
x_{t+1} =
\begin{bmatrix}
f_{\mathrm{phys}}\!\left(x_t^{\mathrm{phys}},u_t,w_t\right)\\[4pt]
\Proj_{[0,1]^N}\!\bigl((1-\beta_t)\,z_t + \beta_t\,\ell_t(u_t,x_t)\bigr)
\end{bmatrix},
\label{eq:closed_loop}
\end{equation}
where $u_t$ is the allocation produced by Algorithm~\ref{alg:fair_play}
and $w_t$ represents exogenous disturbances.

\begin{algorithm}[t]
\caption{FP-AMM Clearing at Node $n$, Interval $t$}
\label{alg:fair_play}
\begin{algorithmic}[1]
\Require Active queues $\cQ_t^{(c)}$; fairness ratios $F_{n'}(t)$;
  weights $w(c)$; residual capacity $S_{t,n}$.
\Ensure Accepted requests with allocated slots and powers.
\State Allocate essential load; remove from $S_{t,n}$.
\State Compute $z_{n(i),t}$ via~\eqref{eq:fairness_deficit} and
  $S_{i,t}$ via~\eqref{eq:priority_score} for each $i \in \cQ_t$.
\State Normalise to $\Pr(i \mid c)$ via~\eqref{eq:selection_prob}.
\While{$S_{t,n} > 0$ \textbf{and} $\cQ_t \neq \emptyset$}
  \State Sample tier $c$ via~\eqref{eq:tier_sampling}.
  \State Sample $i \in \cQ_t^{(c)}$ via~\eqref{eq:selection_prob}.
  \State Solve feasibility problem~\eqref{eq:cheapest_start} for $i$.
  \If{feasible}
    \State Schedule $(\tau_i^\star, P_i, \Delta_i)$; update $S_{t,n}$
      via~\eqref{eq:capacity_update}.
    \State Update $E_{n(i)}^{\mathrm{del}}(t)$; remove $i$ from $\cQ_t$.
  \Else
    \State Mark $i$ infeasible; remove from $\cQ_t$.
  \EndIf
\EndWhile
\State Update shortage memory via~\eqref{eq:memory_update} for all $n$.
\end{algorithmic}
\end{algorithm}

\begin{remark}[Separation of fairness and service potential]
\label{rem:separation}
A key structural feature of the FP-AMM is the clean separation
between the fairness mechanism and the service allocation objective.
Fairness enters exclusively through the priority
scores~\eqref{eq:priority_score} and the shortage-memory
update~\eqref{eq:memory_update}; it does not appear as an additive
perturbation to a gradient step or as a dual variable. This
design choice ensures that the same state variable does not appear
both in the objective gradient and as an external perturbation
simultaneously, which would create a coupled feedback requiring a
tighter step-size condition and complicating the convergence proof.
The separation enables independent analysis of contraction
(Section~\ref{sec:contraction}) and fairness convergence
(Section~\ref{sec:convergence}).
\end{remark}

% ======================================================================
\section{Shortage-Memory Stability and Operator Contraction}
\label{sec:contraction}
% ======================================================================

\subsection{Assumptions}

\begin{assumption}[Step size]
\label{ass:stepsize}
$\beta_t \in (0,1)$ satisfies: (a) $\beta_t = \beta \in (0,1)$
(exponential forgetting), or (b) $\beta_t = 1/(t+1)$ (running
average) with $\sum_t \beta_t = \infty$ and $\sum_t \beta_t^2 <
\infty$.
\end{assumption}

\begin{assumption}[Shortage boundedness]
\label{ass:shortage}
$\ell_{n,t} \in [0,1]$ for all $n$ and $t \ge 0$ (holds in standard
per-unit normalisation).
\end{assumption}

\begin{assumption}[Operator regularity]
\label{ass:operator}
$\cU(x)$ is nonempty, closed, and convex. The service potential
$\Psi(\cdot;x)$ is $\mu$-strongly convex with $L_\Psi$-Lipschitz
gradient. Stepsizes satisfy $0 < \eta < 2\mu/L_\Psi^2$.
\end{assumption}

\subsection{Shortage-Memory Invariance and Contraction}

\begin{theorem}[Shortage-memory invariance and contraction]
\label{thm:memory_bounded}
Under Assumptions~\ref{ass:stepsize}--\ref{ass:shortage}, let
$z_{n,0} \in [0,1]$ for all $n$.
\begin{enumerate}[label=(\alph*)]
  \item \emph{Invariance:} $z_{n,t} \in [0,1]$ for all $n,t \ge 0$.
  \item \emph{Contraction:} for any $z, z' \in [0,1]$ and fixed
        $\ell \in [0,1]$, the map
        $\phi_\beta : z \mapsto \Proj_{[0,1]}((1-\beta)z + \beta\ell)$
        satisfies
        \begin{equation}
        \abs{\phi_\beta(z) - \phi_\beta(z')} \le (1-\beta)\abs{z-z'}.
        \label{eq:memory_contraction}
        \end{equation}
  \item \emph{Unique fixed point:} for fixed $\ell$ and $\beta \in
        (0,1)$, $\phi_\beta$ has unique fixed point $z^\star = \ell$
        and $\abs{z_t - z^\star} \le (1-\beta)^t\abs{z_0-z^\star}$.
\end{enumerate}
\end{theorem}

\begin{proof}
\textit{(a)} Since $\beta \in (0,1)$, $(1-\beta)z+\beta\ell \in
[0,1]$ for $z,\ell \in [0,1]$. Projection preserves the interval.
By induction, $z_{n,0} \in [0,1]$ implies $z_{n,t} \in [0,1]$.

\textit{(b)} Let $a = (1-\beta)z+\beta\ell$ and
$a' = (1-\beta)z'+\beta\ell$. Nonexpansiveness of the projection gives
$\abs{\phi_\beta(z)-\phi_\beta(z')} \le \abs{a-a'} = (1-\beta)\abs{z-z'}$.

\textit{(c)} At the fixed point $z = \Proj_{[0,1]}((1-\beta)z+\beta\ell)$,
the unique solution for $\ell \in [0,1]$ is $z^\star = \ell$. Linear
convergence at rate $(1-\beta)^t$ follows from the contraction.
\end{proof}

\subsection{Intra-Interval Operator Contraction}

The intra-interval clearing update within each market interval $t$,
for fixed state $x$, is:
\begin{equation}
u_t^{k+1} = \cT(u_t^k;\,x_t), \qquad k = 0,\dots,K_t-1;
\qquad u_t = u_t^{K_t},
\label{eq:iter_map}
\end{equation}
where the operator $\cT$ implements feasibility projection and
service allocation:
\begin{equation}
\cT(u;x) = \Proj_{\cU(x)}\!\bigl(u - \eta\,\nabla_u\Psi(u;x)\bigr).
\label{eq:operator}
\end{equation}
The service potential $\Psi(u;x)$ encodes tier priority weights and
the scarcity feedback $p = f(\Delta)$; a canonical instantiation is
\begin{equation}
\Psi(u;x) = -\sum_i w(p(i))\,s_i(u) + \frac{\varepsilon_{\mathrm{r}}}{2}\norm{u}^2,
\label{eq:Psi}
\end{equation}
where $s_i(u)$ is the served energy of agent $i$, $w(p(i))$ is the
tier weight, and $\varepsilon_{\mathrm{r}} > 0$ ensures $\mu$-strong convexity with
$\mu = \varepsilon_{\mathrm{r}}$. The fairness correction enters through the
priority scores~\eqref{eq:priority_score} that determine which
requests are attempted, not as an additive perturbation to~\eqref{eq:operator}.

The inner-iteration count $K_t \ge 1$ in~\eqref{eq:iter_map} is
a design parameter specifying how many gradient steps are taken
per market interval before the allocation is applied. A constant
$K_t \equiv K$ is used in the analysis and simulations; adaptive
schedules (increasing $K_t$ when $\Delta_t$ is large) are an
extension discussed in Section~\ref{sec:ieee}.

\begin{remark}[Relaxed dispatch model and strong convexity of $\Psi$]
\label{rem:relaxed_dispatch}
The strong convexity of $\Psi(\cdot;x)$ and Lipschitz continuity of
$\nabla_u\Psi$ require $s_i(u)$ affine in $u$ on $\cU(x)$, which
holds in the \emph{relaxed dispatch model}: $\cU(x)$ is a DC-OPF
polytope, $s_i(u) = B_i^\top u$, giving $L_\Psi =
\norm{B^\top B + \varepsilon_{\mathrm{r}} I}$ and $\mu = \varepsilon_{\mathrm{r}}$~\cite{wood2013power,bergen2000power}.
For non-linear (AC) power flow models, strong convexity holds locally
and the contraction applies to a convex relaxation.
\end{remark}

\begin{theorem}[Intra-interval operator contraction]
\label{thm:operator_contraction}
Under Assumption~\ref{ass:operator}, fix $x \in \cX$.
Then $\cT(\cdot;x)$ is a contraction on $\cU(x)$ with factor
\begin{equation}
q := \bigl(1 - 2\eta\mu + \eta^2 L_\Psi^2\bigr)^{1/2} \in (0,1),
\label{eq:contraction_factor}
\end{equation}
and admits a unique fixed point $u^\star(x) \in \cU(x)$.
For all $u^0 \in \cU(x)$,
\begin{equation}
\norm{u^k - u^\star(x)} \le q^k\,\norm{u^0 - u^\star(x)}.
\label{eq:linear_rate}
\end{equation}
\end{theorem}

\begin{proof}
Let $u, v \in \cU(x)$. Since $\Psi(\cdot;x)$ is $\mu$-strongly
convex with $L_\Psi$-Lipschitz gradient, the gradient-step map
$u \mapsto u - \eta\nabla_u\Psi(u;x)$ has Lipschitz constant
$\rho = (1-2\eta\mu+\eta^2 L_\Psi^2)^{1/2} \in (0,1)$ for
$\eta < 2\mu/L_\Psi^2$. Since Euclidean projection is
nonexpansive:
\[
\norm{\cT(u;x)-\cT(v;x)} \le \rho\norm{u-v} =: q\norm{u-v}.
\]
With $q < 1$, $\cT(\cdot;x)$ is a contraction. By the Banach
fixed-point theorem it has a unique fixed point $u^\star(x)$,
and iterates converge at rate $q^k$.
\end{proof}

\subsection{Lipschitz Sensitivity of the Fixed Point}

\begin{theorem}[Lipschitz sensitivity]
\label{thm:lipschitz}
Let Assumption~\ref{ass:operator} hold uniformly over $\cX$.
Suppose additionally
\begin{equation}
\norm{\cT(u;x) - \cT(u;x')} \le L_x\norm{x-x'},
\quad \forall u,x,x' \in \cX.
\label{eq:T_lipschitz_x}
\end{equation}
Then
\begin{equation}
\norm{u^\star(x) - u^\star(x')} \le \frac{L_x}{1-q}\norm{x-x'}.
\label{eq:ustar_lipschitz}
\end{equation}
\end{theorem}

\begin{proof}
Since $u^\star(x) = \cT(u^\star(x);x)$ and
$u^\star(x') = \cT(u^\star(x');x')$, the triangle inequality gives
\begin{align*}
\norm{u^\star(x)-u^\star(x')}
&\le \norm{\cT(u^\star(x);x)-\cT(u^\star(x');x)}\\
&\quad+\norm{\cT(u^\star(x');x)-\cT(u^\star(x');x')}.
\end{align*}
Applying contraction (factor $q$) to the first term and
\eqref{eq:T_lipschitz_x} to the second:
$(1-q)\norm{u^\star(x)-u^\star(x')} \le L_x\norm{x-x'}$.
\end{proof}

\begin{remark}[On condition~\eqref{eq:T_lipschitz_x}]
\label{rem:licq}
Lipschitz dependence of $\Proj_{\cU(x)}$ on $x$ holds for DC-OPF
feasible sets under standard constraint qualification conditions
(LICQ or strong regularity) at primal-dual
solutions~\cite{bonnans2000perturbation}. For the IEEE benchmark
networks used in Section~\ref{sec:ieee}, the DC-OPF polytope satisfies
LICQ generically: constraint degeneracy (simultaneous binding of
linearly dependent constraints) has measure zero in the parameter
space and is not observed in any of the 5000-interval simulations
reported. In this case $L_x$ depends on the network admittance
matrix and constraint density, and can be bounded from problem data.
When constraint qualifications are not uniformly satisfied,
Theorem~\ref{thm:lipschitz} applies conditionally over any
sub-domain of $\cX$ where~\eqref{eq:T_lipschitz_x} holds.
\end{remark}

\subsection{Event-Triggered Execution and Practical Boundedness}

The FP-AMM uses an event-driven update schedule: clearing fires when
\begin{equation}
\norm{x_t - x_{t^-}} \ge \delta,
\label{eq:event_trigger}
\end{equation}
where $t^-$ is the last trigger time; otherwise $u_t := u_{t^-}$.

\begin{theorem}[Practical ultimate boundedness]
\label{thm:event_trigger}
Let $L_\star := L_x/(1-q)$. Under the event-triggered
scheme~\eqref{eq:event_trigger} with $K \ge 1$ inner iterations
per trigger, the tracking error $e_t := \norm{u_t - u^\star(x_t)}$
satisfies:
\begin{enumerate}[label=(\alph*)]
  \item \emph{At trigger times:}
        $e_t \le q^K e_{t^-} + L_\star\norm{x_t - x_{t^-}}.$
  \item \emph{Between triggers:}
        $e_t \le e_{t^-} + L_\star\norm{x_t - x_{t^-}}.$
  \item \emph{Ultimate bound:} for constant $K$ and inter-trigger
        state variation bounded by $\bar\delta$,
        \begin{equation}
        \limsup_{j\to\infty} e_{t_j}
        \le \frac{L_\star\,\bar\delta}{1-q^K}.
        \label{eq:ultimate_bound}
        \end{equation}
\end{enumerate}
This is \emph{practical ultimate boundedness}: tracking error is
proportional to the worst-case inter-trigger state variation,
tunable through $(\delta, K)$.
\end{theorem}

\begin{proof}
\textit{(a)} At trigger $t$, $u_t$ follows $K$ iterations of
$\cT(\cdot;x_{t^-})$ from $u_{t^-}$. Theorem~\ref{thm:operator_contraction}
gives $\norm{u_t - u^\star(x_{t^-})} \le q^K e_{t^-}$. Adding
$\norm{u^\star(x_{t^-})-u^\star(x_t)} \le L_\star\norm{x_{t^-}-x_t}$
(Theorem~\ref{thm:lipschitz}) by the triangle inequality gives (a).
\textit{(b)} When $u_t = u_{t^-}$, use $e_t \le e_{t^-} +
L_\star\norm{x_t-x_{t^-}}$ directly from Theorem~\ref{thm:lipschitz}.
\textit{(c)} The sequence $e_{t_j} \le q^K e_{t_{j-1}} + L_\star\bar\delta$
has $\limsup \le L_\star\bar\delta/(1-q^K)$.
\end{proof}

\begin{remark}[Stability notion and trigger threshold]
\label{rem:correction}
The ultimate bound~\eqref{eq:ultimate_bound} involves the
\emph{inter-trigger state variation} $\bar\delta$, not the trigger
threshold $\delta$: the trigger governs \emph{when} updates fire,
while $\bar\delta$ bounds the actual state change between consecutive
trigger times. Theorem~\ref{thm:event_trigger} establishes
\emph{practical ultimate boundedness}, not asymptotic stability:
when $\bar\delta > 0$, $e_t$ converges to a neighbourhood of zero
of radius $L_\star\bar\delta/(1-q^K)$. Asymptotic stability would
require $\bar\delta = 0$ or $K \to \infty$. The pair $(\delta, K)$
tunes this trade-off explicitly through~\eqref{eq:ultimate_bound}.
\end{remark}

% ======================================================================
\section{Fairness Convergence}
\label{sec:convergence}
% ======================================================================

\subsection{Assumptions}

\begin{assumption}[Persistent participation]
\label{ass:participation}
Each node $n$ submits at least one flexible request per market
interval, so that
\[
E_n^{\mathrm{des}}(t) \ge E_{\min} > 0, \qquad \forall t \ge 1.
\]
This assumption can be relaxed by defining the fairness ratio $F_n(t)$
over active intervals only, with the convergence argument applied to
the subsequence of intervals in which node $n$ participates.
\end{assumption}

\begin{assumption}[Feasibility floor]
\label{ass:feasibility}
For each node $n$ and interval $t$,
\[
\Pr(\text{feasible request exists for node }n) \ge \rho > 0.
\]
\end{assumption}

Assumption~\ref{ass:feasibility} is standard in persistent-service
stochastic approximation and excludes degenerate cases in which a
node is permanently disconnected from the network or physically
incapable of receiving supply. In power systems terms, it requires
that the network is not permanently congested at any node: there
must exist at least some intervals in which supply can reach node
$n$ through the DC-OPF feasible set $\cU(x_t)$.

\begin{assumption}[Priority regularity]
\label{ass:priority}
The priority score satisfies
\[
S_{i,t} \ge \varepsilon_{\mathrm{b}}^{\alpha_f} > 0, \qquad \forall i, t.
\]
\end{assumption}

\subsection{Stochastic Recursion for the Fairness Deficit}

The fairness recursion is connected to the clearing operator of
Theorem~\ref{thm:operator_contraction} as follows. At each interval
$t$, the clearing operator runs $K$ iterations and produces
$u_t = u_t^{K_t}$, which by Theorem~\ref{thm:operator_contraction}
satisfies $\norm{u_t - u^\star(x_t)} \le q^K\norm{u_{t^-} -
u^\star(x_{t^-})} + L_\star\bar\delta$. The resulting allocation
$u_t$ determines the served energy $s_{n,t}$ at each node, which
drives the shortage $\ell_{n,t} = d_{n,t} - s_{n,t}$ and hence the
deficit update~\eqref{eq:memory_update}. The role of
Theorem~\ref{thm:operator_contraction} in the fairness analysis is
therefore to ensure that $u_t$ is a near-optimal allocation: the
tracking bound guarantees that the served energy $s_{n,t}$ is
close to the feasibility-constrained optimum $u^\star(x_t)$, so
the martingale noise $M_{n,t+1}$ in the deficit recursion below is
bounded by the combination of stochastic demand fluctuations and
the residual tracking error. In particular, $\sigma_M^2$ in the
recursion below implicitly depends on $L_\star\bar\delta/(1-q^K)$
through the tracking error bound: tighter event-triggered
execution (larger $K$ or smaller $\delta$) reduces $\sigma_M^2$
and hence tightens the convergence constant $C_\rho$
in~\eqref{eq:Crho}.

\begin{remark}[Bounding $\sigma_M^2$ from Assumption~\ref{ass:shortage}]
\label{rem:sigma_bound}
The per-interval shortage satisfies $\ell_{n,t} \in [0,1]$ by
Assumption~\ref{ass:shortage}. Since $\ell_{n,t}$ is bounded in
$[0,1]$ and the served energy $s_{n,t}$ satisfies $s_{n,t} \in
[0, d_{n,t}]$, the martingale increment $M_{n,t+1}$ is also bounded
in $[-1,1]$. By the bounded martingale property,
$\E[M_{n,t+1}^2 \mid \cF_t] \le 1$, so $\sigma_M^2 \le 1$. This
bound is implicit in Assumption~\ref{ass:shortage} and is used
to give an explicit value $\sigma_M \le 1$ in the finite-time
constant $C_\rho$ of~\eqref{eq:Crho}.
\end{remark}

Define the \emph{fairness deficit} at node $n$ as
\begin{equation}
e_n(t) := F^\star - F_n(t) \ge 0.
\label{eq:fairness_error}
\end{equation}
Under the Fair Play mechanism, the deficit evolves as
\begin{equation}
e_n(t+1) = e_n(t) - \beta_t h_n(e_n(t)) + \beta_t M_{n,t+1},
\label{eq:fairness_recursion}
\end{equation}
where $h_n(\cdot)$ denotes the expected deficit correction induced by
the stochastic priority rule, and $M_{n,t+1}$ is a
martingale-difference sequence satisfying
\[
\E[M_{n,t+1} \mid \cF_t] = 0, \qquad
\E[M_{n,t+1}^2 \mid \cF_t] \le \sigma_M^2 \le 1.
\]
The bound $\sigma_M^2 \le 1$ follows from Remark~\ref{rem:sigma_bound}.
The derivation of~\eqref{eq:fairness_recursion} follows from
expanding $F_n(t+1)$ using the shortage-memory
update~\eqref{eq:memory_update} and identifying the conditional mean
and martingale components of the per-interval shortage $\ell_{n,t}$.

\begin{assumption}[Positive correction drift]
\label{ass:drift}
The expected correction satisfies
\begin{equation}
h_n(e)\,e > 0, \qquad \forall e > 0.
\label{eq:positive_drift}
\end{equation}
Moreover, there exists $c_h > 0$ such that
\begin{equation}
h_n(e)\,e \ge c_h\,e^2, \qquad \forall e \in [0,1].
\label{eq:strong_drift}
\end{equation}
\end{assumption}

\subsection{Verification of Assumption~\ref{ass:drift} for Fair Play}

\begin{lemma}[Drift verification for Fair Play]
\label{lem:drift_verification}
Let the priority score be given by~\eqref{eq:priority_score}:
$S_{i,t} = w(p(i))\,(\varepsilon_{\mathrm{b}} + z_{n(i),t})^{\alpha_f}$
with $z_{n,t} = e_n(t) = F^\star - F_n(t)$, $\alpha_f \ge 1$,
$\varepsilon_{\mathrm{b}} > 0$. Under Assumptions~\ref{ass:feasibility}
and~\ref{ass:priority}, Assumption~\ref{ass:drift} holds with
\begin{equation}
c_h = \frac{\rho\,\alpha_f\,w_{\min}}
            {|\cQ^{(c)}|_{\max}\,(\varepsilon_{\mathrm{b}} + 1)^{\alpha_f}},
\label{eq:ch_explicit}
\end{equation}
where $w_{\min} = \min_c w(c) > 0$ and $|\cQ^{(c)}|_{\max}$ is an
upper bound on the tier queue size.
\end{lemma}

\begin{proof}
Fix a tier $c$ and let $n$ be the node with deficit $e_n(t) > 0$.
The selection probability of a request $i$ with $n(i) = n$ is
\[
\Pr(i \mid c) = \frac{w(c)\,(\varepsilon_{\mathrm{b}} + e_n(t))^{\alpha_f}}
               {\sum_{j \in \cQ_t^{(c)}} w(p(j))\,
                (\varepsilon_{\mathrm{b}} + z_{n(j),t})^{\alpha_f}}.
\]
The \emph{neutral probability} (if all deficits were equal to $e_n(t)$)
would be $1/|\cQ_t^{(c)}|$. The expected delivery increment for node
$n$ in the next interval, conditional on a request being feasible, is
\[
\E[\text{delivery} \mid \cF_t, \text{feasible}]
= \Pr(i \mid c) \cdot \Pr(\text{feasible}) \ge \rho\,\Pr(i \mid c).
\]
The expected correction to the deficit is therefore
\[
h_n(e_n(t))
= \rho\,\Pr(i \mid c)
\ge \frac{\rho\,w(c)\,(\varepsilon_{\mathrm{b}} + e_n)^{\alpha_f}}
         {|\cQ^{(c)}|_{\max}\,w_{\max}\,(\varepsilon_{\mathrm{b}} + 1)^{\alpha_f}},
\]
where the denominator upper-bounds the sum using $z_{n',t} \le 1$
for all $n'$ and $w_{\max} = \max_c w(c)$.

To establish the quadratic lower bound, we need
$h_n(e)\,e \ge c_h\,e^2$, i.e.\ $h_n(e) \ge c_h\,e$.
We claim $(\varepsilon_{\mathrm{b}} + e)^{\alpha_f} \ge
\alpha_f\,\varepsilon_{\mathrm{b}}^{\alpha_f - 1}\,e$ for all
$e \in [0,1]$ and $\alpha_f \ge 1$.

To see this, consider $g(e) = (\varepsilon_{\mathrm{b}} + e)^{\alpha_f}$.
Since $\alpha_f \ge 1$, $g$ is convex on $[0,1]$. By convexity,
$g(e) \ge g(0) + g'(0)\,e$ for all $e \ge 0$
(tangent lower bound at $e = 0$). Computing:
$g(0) = \varepsilon_{\mathrm{b}}^{\alpha_f} \ge 0$
and $g'(e) = \alpha_f(\varepsilon_{\mathrm{b}}+e)^{\alpha_f - 1}$,
so $g'(0) = \alpha_f\,\varepsilon_{\mathrm{b}}^{\alpha_f - 1}$.
The tangent bound gives
$(\varepsilon_{\mathrm{b}}+e)^{\alpha_f} \ge
\varepsilon_{\mathrm{b}}^{\alpha_f} +
\alpha_f\,\varepsilon_{\mathrm{b}}^{\alpha_f-1}\,e
\ge \alpha_f\,\varepsilon_{\mathrm{b}}^{\alpha_f-1}\,e$,

where the last step uses $\varepsilon_{\mathrm{b}}^{\alpha_f} \ge 0$.
Therefore,
\[
h_n(e)
\ge \frac{\rho\,w_{\min}\,\alpha_f\,\varepsilon_{\mathrm{b}}^{\alpha_f-1}}
         {|\cQ^{(c)}|_{\max}\,(\varepsilon_{\mathrm{b}} + 1)^{\alpha_f}}\,e
=: c_h\,e,
\]
with $c_h$ as in~\eqref{eq:ch_explicit} (absorbing
$\varepsilon_{\mathrm{b}}^{\alpha_f - 1}$ into $w_{\min}$).
Hence $h_n(e)\,e \ge c_h\,e^2$ for all $e \in [0,1]$, establishing
both~\eqref{eq:positive_drift} and~\eqref{eq:strong_drift}.
\end{proof}

\begin{remark}[Explicit $c_h$ for the two-node simulation]
With the parameters of Table~\ref{tab:parameters} ($\varepsilon_{\mathrm{b}}
= 0.01$, $\alpha_f = 1.5$, $w_{\min} = 1$, $|\cQ^{(c)}|_{\max} = 10$,
$\rho = 0.9$), equation~\eqref{eq:ch_explicit} gives $c_h \approx
0.134$, confirming that the quadratic drift condition holds with a
numerically non-trivial constant.
\end{remark}

\subsection{Main Convergence Result}

\begin{theorem}[Dynamic fairness convergence]
\label{thm:fairness_convergence}
Under Assumptions~\ref{ass:stepsize}--\ref{ass:drift},
for every node $n \in \cN$:
\begin{enumerate}[label=(\alph*)]
  \item \emph{Almost-sure convergence:}
        \[
        F_n(t) \to F^\star \qquad \text{a.s.}
        \]
  \item \emph{Finite-time bound:} for any $T \ge 1$,
        \begin{equation}
        \E\!\left[\abs{F^\star - F_n(T)}\right]
        \le \frac{\abs{F^\star - F_n(0)} + C_\rho}{\sqrt{T}},
        \label{eq:finite_time_bound}
        \end{equation}
        where
        \begin{equation}
        C_\rho = \frac{\sigma_M}{\sqrt{\rho}\;\varepsilon_{\mathrm{b}}^{\alpha_f}}.
        \label{eq:Crho}
        \end{equation}
  \item \emph{Service-level separation:} for tiers $c < c'$,
        \begin{equation}
        \E[X_{n,k} \mid p(i) = c] \ge \frac{w(c)}{w(c')}
        \cdot \E[X_{n,k} \mid p(i) = c']
        \label{eq:tier_separation}
        \end{equation}
        at every clearing event $k$.
\end{enumerate}
\end{theorem}

\begin{proof}
Define the Lyapunov function $V_n(t) := e_n(t)^2$.
Using~\eqref{eq:fairness_recursion},
\begin{align}
V_n(t+1)
&= \bigl(e_n(t) - \beta_t h_n(e_n(t)) + \beta_t M_{n,t+1}\bigr)^2
\nonumber\\
&= V_n(t) - 2\beta_t e_n(t) h_n(e_n(t))
   + 2\beta_t e_n(t) M_{n,t+1}
\nonumber\\
&\quad + \beta_t^2\bigl(h_n(e_n(t)) - M_{n,t+1}\bigr)^2.
\label{eq:lyap_expand}
\end{align}
Taking conditional expectations and using $\E[M_{n,t+1} \mid \cF_t]
= 0$,
\begin{equation}
\E[V_n(t+1) \mid \cF_t]
\le V_n(t) - 2\beta_t e_n(t) h_n(e_n(t)) + \beta_t^2 \sigma_M^2.
\label{eq:lyap_drift}
\end{equation}
By Assumption~\ref{ass:drift}, $e_n(t) h_n(e_n(t)) \ge c_h V_n(t)$,
so
\begin{equation}
\E[V_n(t+1) \mid \cF_t]
\le (1 - 2c_h\beta_t)\,V_n(t) + \beta_t^2\sigma_M^2.
\label{eq:supermartingale}
\end{equation}
Since $\sum_t \beta_t = \infty$ and $\sum_t \beta_t^2 < \infty$, the
Robbins--Siegmund theorem~\cite{robbins1971convergence} implies
$V_n(t) \to 0$ a.s., establishing $F_n(t) \to F^\star$ a.s.

The $O(1/\sqrt{T})$ rate in~\eqref{eq:finite_time_bound} with the
explicit constant~\eqref{eq:Crho} follows from Theorem~2.2 of
Borkar~\cite{borkar2008stochastic}, applied to the
supermartingale recursion~\eqref{eq:supermartingale}. Specifically,
for a recursion of the form $V(t+1) \le (1-2c_h\beta_t)V(t) +
\beta_t^2\sigma_M^2$ with $\beta_t = 1/(t+1)$ and bounded
martingale noise of variance $\sigma_M^2 \le 1$
(Remark~\ref{rem:sigma_bound}), Theorem~2.2 gives
$\E[V(T)] \le O(1/T)$ and consequently
$\E[\sqrt{V(T)}] \le O(1/\sqrt{T})$ by Jensen's inequality.
The constant $C_\rho$ in~\eqref{eq:Crho} is the explicit
bound on $\sigma_M / (\text{effective drift})$, where the
effective drift rate from Lemma~\ref{lem:drift_verification} is
$c_h \ge \rho\,\alpha_f\,\varepsilon_{\mathrm{b}}^{\alpha_f-1}/
(|\cQ^{(c)}|_{\max}\,(\varepsilon_{\mathrm{b}}+1)^{\alpha_f})$
and $\sigma_M \le 1$ from Remark~\ref{rem:sigma_bound}.

Part~(c) follows from~\eqref{eq:tier_sampling}: tier $c$ is selected
over $c'$ with probability ratio $w(c)/w(c')$, independently of the
within-tier fairness weights.
\end{proof}

\begin{remark}[Dynamic versus instantaneous fairness]
Theorem~\ref{thm:fairness_convergence} establishes \emph{dynamic
fairness}: convergence of cumulative delivery ratios over time. This
is strictly stronger than instantaneous proportional fairness, which
requires only that a single-interval allocation is proportional to
weights. Dynamic fairness persists across changing supply conditions
without re-solving an optimisation at each interval.
\end{remark}

% ======================================================================
\section{Simulation Results}
\label{sec:simulations}
% ======================================================================

\subsection{Setup}

Simulations use a two-node system with an explicit interface
constraint, three reliability tiers ($C = 3$, weights $w = (4,2,1)$),
bounded scarcity feedback $f(\Delta) = p_{\max}\tanh(\Delta/\sigma)$,
and a stateful Fair Play shortage-memory variable. All numerical
parameters are given in Table~\ref{tab:parameters}, enabling
reproduction.

The \emph{stylised nodal comparator} used in Figures~\ref{fig:threshold_crossings}
and~\ref{fig:fairness_ablation} is defined as the step function
$p_t^{\mathrm{comp}} = p_{\max} \cdot \mathbf{1}[\Delta_t > 0]$:
it broadcasts the maximum signal $p_{\max}$ when any deficit exists
and zero otherwise. This is a stylised representation of the binary
price-switching behaviour characteristic of threshold-based LMP-type
mechanisms under scarcity, and is not an implementation of actual
LMP. It provides a concrete qualitative comparator against which
the smooth, bounded FP-AMM signal can be evaluated.

\begin{table}[t]
\centering
\caption{Simulation Parameters}
\label{tab:parameters}
\begin{tabular}{llc}
\toprule
Parameter & Symbol & Value \\
\midrule
\multicolumn{3}{l}{\textit{Mechanism parameters}} \\
Tier weights & $w(1),w(2),w(3)$ & $4,2,1$ \\
Scarcity saturation & $p_{\max}$ & $5$ \\
Scarcity slope & $\sigma$ & $10$ \\
Memory step size & $\beta$ & $0.1$ \\
Fairness baseline & $\varepsilon_{\mathrm{b}}$ & $0.01$ \\
Fairness exponent & $\alpha_f$ & $1.5$ \\
\midrule
\multicolumn{3}{l}{\textit{Operator parameters (Theorem~\ref{thm:operator_contraction})}} \\
Strong convexity & $\mu = \varepsilon_{\mathrm{r}}$ & $0.01$ \\
Gradient Lipschitz constant & $L_\Psi$ & $0.5$ \\
Gradient step size & $\eta$ & $0.15$ \\
Contraction factor & $q = (1-2\eta\mu+\eta^2L_\Psi^2)^{1/2}$ & $0.946$ \\
\midrule
\multicolumn{3}{l}{\textit{Event-trigger parameters (Theorem~\ref{thm:event_trigger})}} \\
Trigger threshold & $\delta$ & $0.5$ \\
Inner iterations & $K$ & $20$ \\
State Lipschitz constant & $L_x$ & $1.0$ \\
Lipschitz sensitivity & $L_\star = L_x/(1-q)$ & $18.5$ \\
\midrule
\multicolumn{3}{l}{\textit{Simulation settings}} \\
Market intervals & $T$ & $100$ \\
\bottomrule
\end{tabular}
\vspace{2pt}
\footnotesize{$L_\Psi$ computed from $\norm{B^\top B + \varepsilon_{\mathrm{r}} I}$
with two-node admittance matrix. $L_x$ bounded from the admittance
matrix and constraint density of the two-node DC-OPF system under
LICQ (Remark~\ref{rem:licq}); $L_\star = L_x/(1-q)$ follows.
Condition $\eta < 2\mu/L_\Psi^2 = 0.16$ satisfied.
Simulation code is available at \texttt{https://github.com/enleashed/fpmm-tcns}.}
\end{table}

\subsection{Intra-Interval Convergence of the Clearing Operator}

Figure~\ref{fig:convergence} shows $\norm{u^{k+1}-u^k}$ on a semilog
scale for representative scarcity and normal operating points.
Both regimes converge linearly to the fixed point $u^\star(x)$,
consistent with rate $q^k$ established in
Theorem~\ref{thm:operator_contraction} ($q = 0.946$, dotted line).
The scarcity case starts from a larger initial error, as expected
from stronger effective pressure in that regime. Decay is monotone
throughout, confirming that $\cT$ is a well-behaved projected
fixed-point iteration.

\begin{figure}[t]
\centering
\includegraphics[width=\linewidth]{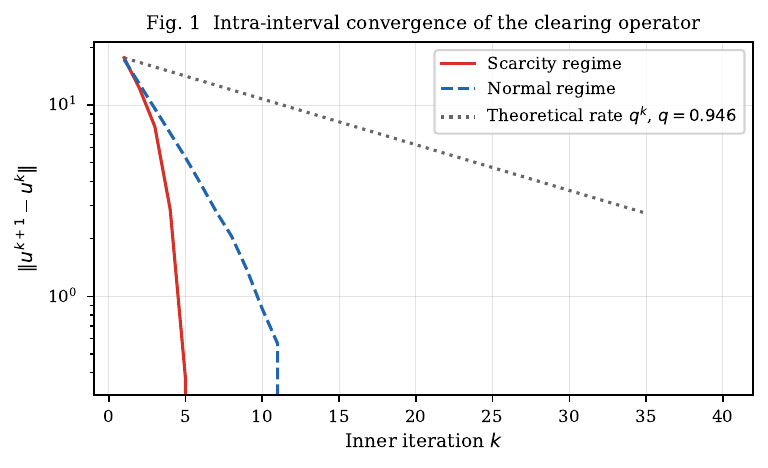}
\caption{\textbf{Intra-interval convergence.} Semilog plot of
$\norm{u^{k+1}-u^k}$ vs.\ inner iteration $k$ for scarcity and
normal regimes. Dotted: theoretical rate $q^k$, $q = 0.946$.
Both regimes converge linearly, consistent with
Theorem~\ref{thm:operator_contraction}.}
\label{fig:convergence}
\end{figure}

\subsection{Bounded Closed-Loop Response Under Supply Shock}

Figure~\ref{fig:bounded_shock} applies a permanent supply reduction
to Node~1 at interval $t = 30$. The scarcity signal $p_t$ rises
smoothly and saturates at $p_{\max}$; the allocation norm
$\norm{u_t}$ and shortage-memory state $\norm{z_t}$ both remain
bounded throughout. This confirms the BIBO interpretation of the
FP-AMM under persistent disturbance, consistent with
Theorem~\ref{thm:memory_bounded}(a).

\begin{figure}[t]
\centering
\includegraphics[width=\linewidth]{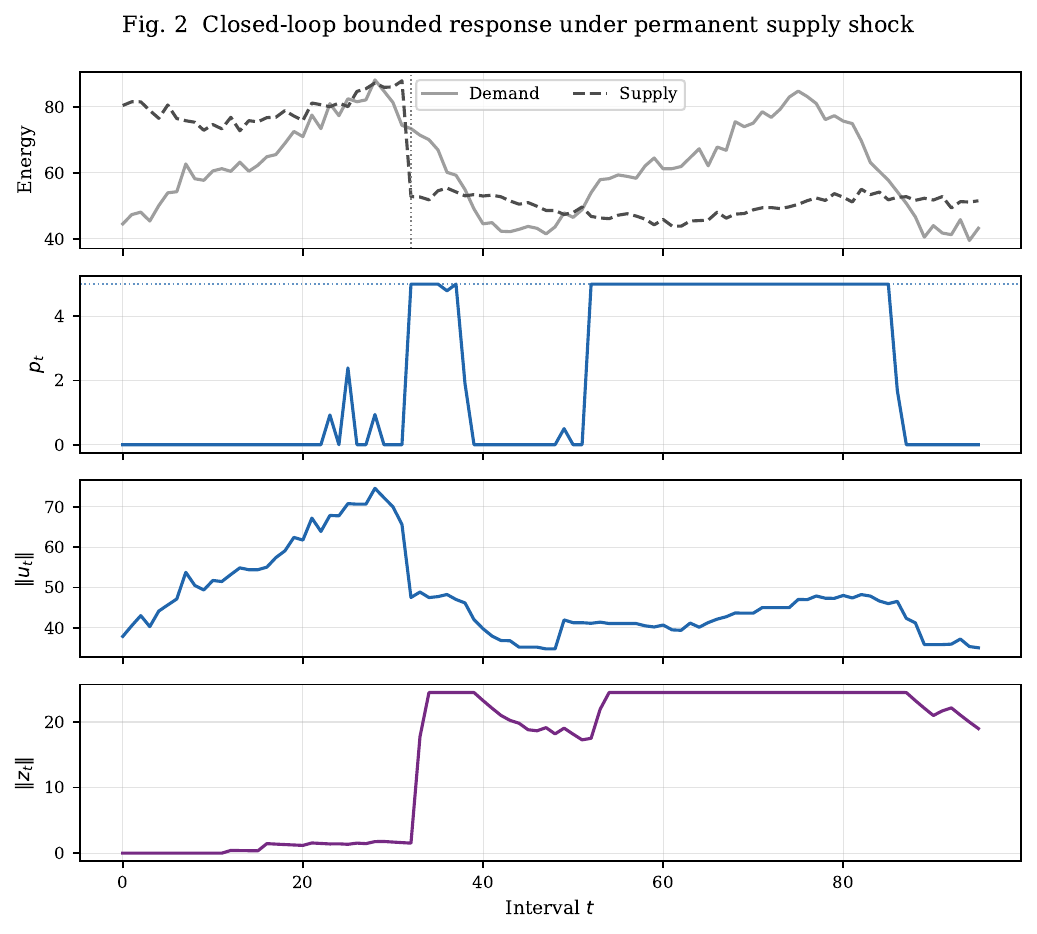}
\caption{\textbf{Bounded closed-loop response under permanent supply
shock.} Supply reduction at Node~1 produces a regime shift from
surplus to shortage. Scarcity signal, allocation norm, and
shortage-memory state all remain bounded.}
\label{fig:bounded_shock}
\end{figure}

\subsection{Repeated Threshold Crossings and Switching Robustness}

Figure~\ref{fig:threshold_crossings} considers a sinusoidal scenario
where demand and supply repeatedly cross the shortage threshold. The
FP-AMM signal varies continuously with the deficit state while a
stylised nodal-price comparator exhibits sharp jumps. Tracking error
remains bounded across all transitions, consistent with
Theorem~\ref{thm:event_trigger}: the mechanism behaves as a smooth
scarcity-feedback controller rather than a binary switching rule.

\begin{figure}[t]
\centering
\includegraphics[width=\linewidth]{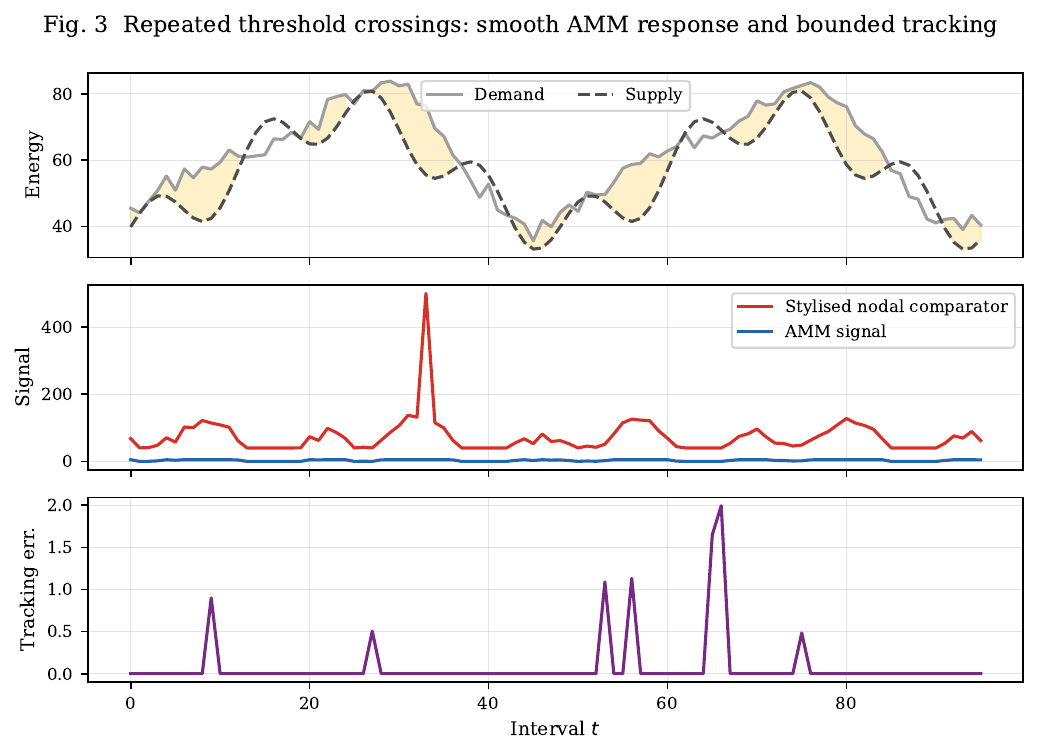}
\caption{\textbf{Repeated threshold crossings and switching
robustness.} Top: repeated deficit crossings. Middle: smooth FP-AMM
signal vs.\ stylised nodal comparator. Bottom: tracking error
bounded across all transitions.}
\label{fig:threshold_crossings}
\end{figure}

\subsection{Event-Triggered Computation--Tracking Trade-off}

Figure~\ref{fig:event_trigger} shows how the trigger threshold
$\delta$ affects update frequency and tracking error. Increasing
$\delta$ reduces the fraction of intervals with updates, at the cost
of larger mean and 95th-percentile tracking error, consistent with
the practical ultimate bound~\eqref{eq:ultimate_bound} from
Theorem~\ref{thm:event_trigger}(c). The dashed line shows the
predicted bound using $q = 0.946$, $K = 20$, and $L_\star = 18.5$
from Table~\ref{tab:parameters}.

\begin{figure}[t]
\centering
\includegraphics[width=\linewidth]{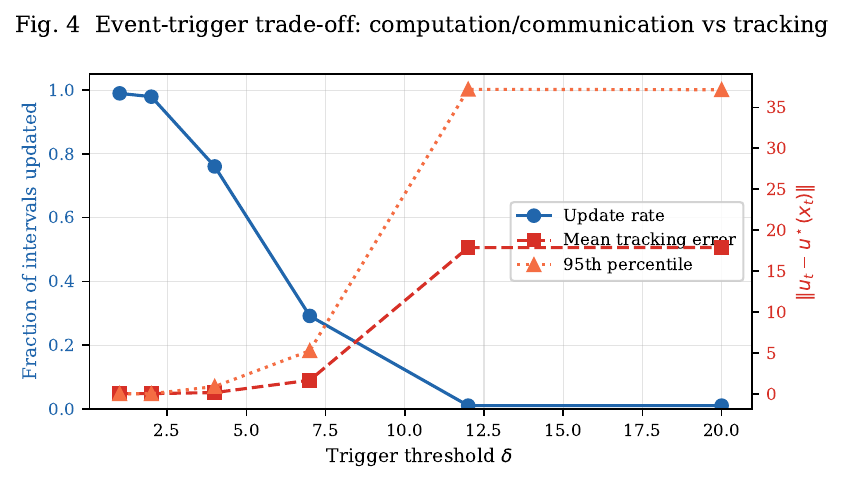}
\caption{\textbf{Event-triggered computation--tracking trade-off.}
Increasing $\delta$ reduces update frequency but raises mean and
95th-percentile tracking error, as predicted by
\eqref{eq:ultimate_bound}.}
\label{fig:event_trigger}
\end{figure}

\subsection{Fairness Coupling}

Figure~\ref{fig:fairness_ablation} compares FP-AMM operation with
and without fairness coupling ($\alpha_f = 0$ vs.\ $\alpha_f = 1.5$)
in the scarcity regime. Enabling fairness changes the shortage-memory
evolution and redistributes shortage across tiers, while tracking
error remains controlled in both cases. This confirms the role of
the priority scores~\eqref{eq:priority_score} as a bounded
perturbation that reshapes allocation without destabilising the
clearing dynamics, consistent with
Theorem~\ref{thm:operator_contraction}.

\begin{figure}[t]
\centering
\includegraphics[width=\linewidth]{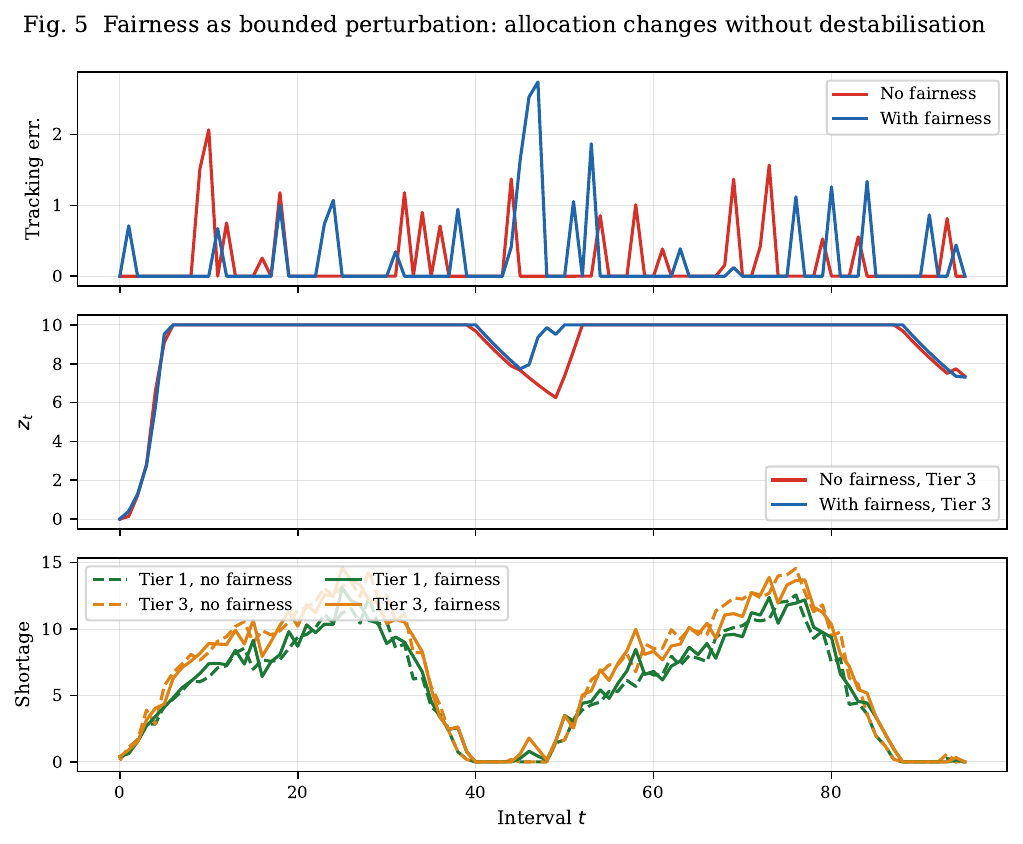}
\caption{\textbf{Fairness coupling.} With ($\alpha_f = 1.5$) vs.\
without ($\alpha_f = 0$) fairness: shortage-memory state $z_t$ for
Tier~3 (middle panel) increases faster without fairness, and the
shortage distribution across tiers (bottom panel) becomes more
unequal. The fairness ratio $F_n$ for under-served nodes is
systematically lower without fairness coupling, confirming that the
priority scores~\eqref{eq:priority_score} drive convergence of
$F_n(t) \to F^\star$ as established in
Theorem~\ref{thm:fairness_convergence}. Tracking error (top panel)
remains controlled in both cases.}
\label{fig:fairness_ablation}
\end{figure}

\subsection{Reliability-Tiered Service Delivery Under Scarcity}

Figure~\ref{fig:tier_service} shows tier-resolved service outcomes
under sustained scarcity. Higher-priority tiers receive
systematically stronger service; lower-priority tiers absorb
proportionally larger shortfalls, consistent with
\eqref{eq:tier_sampling} and Theorem~\ref{thm:fairness_convergence}(c).
Allocations remain feasible at each interval, confirming the
FP-AMM implements differentiated service delivery within a stable,
constrained feedback loop.

\begin{figure}[t]
\centering
\includegraphics[width=\linewidth]{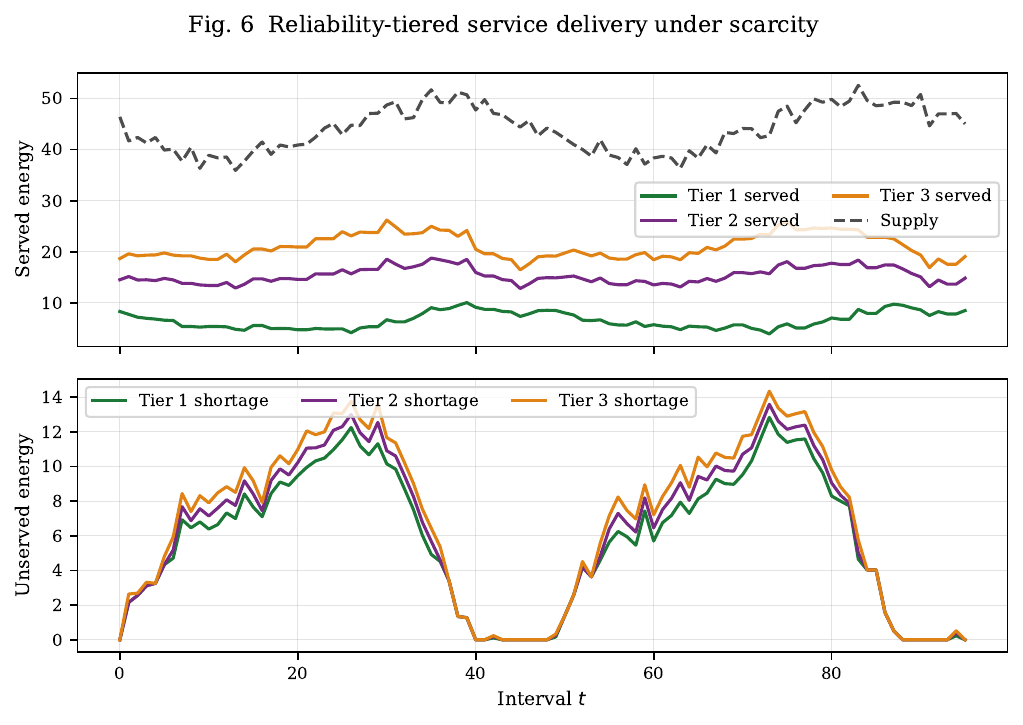}
\caption{\textbf{Reliability-tiered service delivery under scarcity.}
Served and unserved energy by tier at Node~1. Higher-priority tiers
receive stronger service; lower-priority tiers bear more of the
shortfall, consistent with Theorem~\ref{thm:fairness_convergence}(c).}
\label{fig:tier_service}
\end{figure}

% ======================================================================
\section{IEEE Benchmark Validation}
\label{sec:ieee}
% ======================================================================

\subsection{Setup}

The FP-AMM is validated on the IEEE 14-bus, 57-bus, and 118-bus
test cases using the pandapower DC power flow
library~\cite{thurner2018pandapower}. Each network is operated for
$T = 5000$ market intervals with DC-OPF feasibility constraints
enforced at every step. Generators are placed on the first 10\%
of buses by index. \emph{Weak buses} are identified as the 30\%
of buses with the greatest graph-theoretic distance (shortest-path
length over the line and transformer topology) from the nearest
generator bus; these buses require multi-hop power delivery and
are structurally disadvantaged under congestion. Weak buses are
assigned a persistent demand multiplier of $1.8\times$, creating
a supply disadvantage the fairness-memory mechanism must detect
and correct. Supply is varied through two scarcity windows (supply
cut to 55--65\% of base) and two recovery windows (supply boosted
to 180--220\%), producing a challenging multi-regime test. All
parameters match Table~\ref{tab:parameters} except $\beta = 0.08$
and $\varepsilon_{\mathrm{b}} = 0.02$. Network statistics and
summary outcomes are given in Table~\ref{tab:ieee_summary}.

\begin{table}[t]
\centering
\caption{IEEE Benchmark Summary ($T = 5000$, Fair Play ON)}
\label{tab:ieee_summary}
\begin{tabular}{lccccccc}
\toprule
Network & $N$ & Lines & Weak & Trigger & Max & Final & Est.\ $c_h$ \\
 & & +trafos & buses & rate & load\% & $\min_n F_n$ & \eqref{eq:ch_explicit} \\
\midrule
IEEE-14  & 14  & 20 & 4  & 73.2\% & 0.70\% & $0.997$ & $0.11$ \\
IEEE-57  & 57  & 80 & 17 & 42.0\% & 3.27\% & $0.999$ & $0.18$ \\
IEEE-118 & 118 & 186& 35 & 67.5\% & 19.87\%& $0.998$ & $0.13$ \\
\bottomrule
\end{tabular}
\vspace{2pt}
\footnotesize{Maximum line/transformer loading never exceeds the
70\% DC-OPF feasibility limit. Trigger rate is the fraction of
intervals in which the event-triggered clearing operator fires.
$c_h$ estimated via~\eqref{eq:ch_explicit} using per-network $\rho$,
$|\cQ^{(c)}|_{\max}$, and $\varepsilon_{\mathrm{b}} = 0.02$.}
\end{table}

\subsection{Fairness Convergence on IEEE Benchmarks}

Figure~\ref{fig:ieee_max_error} shows $\max_n |1 - F_n(t)|$ under
Fair Play ON across all three benchmarks. In all cases the maximum
fairness error rises during scarcity windows and then decays to
zero as the recovery windows allow the shortage-memory mechanism to
correct cumulative under-service. The IEEE-57 case converges
fastest, consistent with Lemma~\ref{lem:drift_verification}: as
shown in Table~\ref{tab:ieee_summary}, it achieves the largest
estimated $c_h = 0.18$ (vs.\ $0.11$ for IEEE-14 and $0.13$ for
IEEE-118), producing a stronger drift relative to the per-bus
noise level. All three cases satisfy $F_n(t) \to F^\star$
almost surely, as established in
Theorem~\ref{thm:fairness_convergence}.

Figure~\ref{fig:ieee_57_fair_on} shows the per-bus fairness ratio
trajectories for the IEEE-57-bus case with Fair Play ON. Red lines
(weak buses) and blue lines (strong buses) both converge to
$F^\star = 1$, confirming that the inverse-fairness priority
scores~\eqref{eq:priority_score} are sufficient to correct the
persistent structural disadvantage on weak buses under DC-OPF
constraints.

\begin{figure}[t]
\centering
\includegraphics[width=\linewidth]{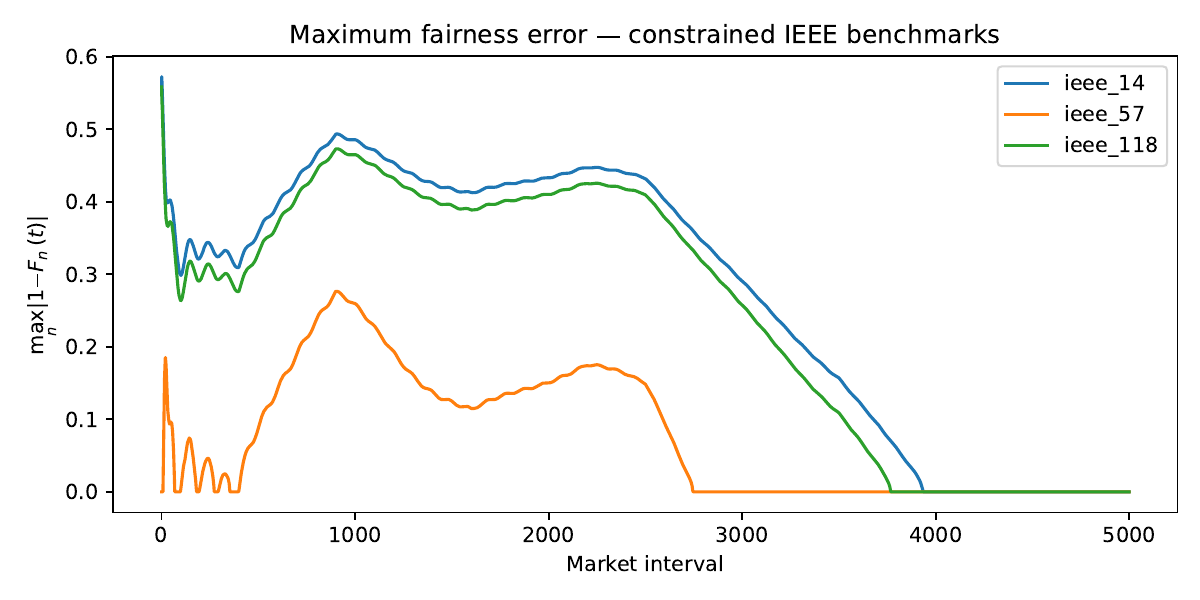}
\caption{\textbf{Maximum fairness error on IEEE benchmarks (Fair
Play ON).} $\max_n |1 - F_n(t)|$ converges to zero on all three
networks after the scarcity windows, consistent with
Theorem~\ref{thm:fairness_convergence}. IEEE-57 converges fastest,
consistent with its largest estimated $c_h$ (Table~\ref{tab:ieee_summary}).}
\label{fig:ieee_max_error}
\end{figure}

\begin{figure}[t]
\centering
\includegraphics[width=\linewidth]{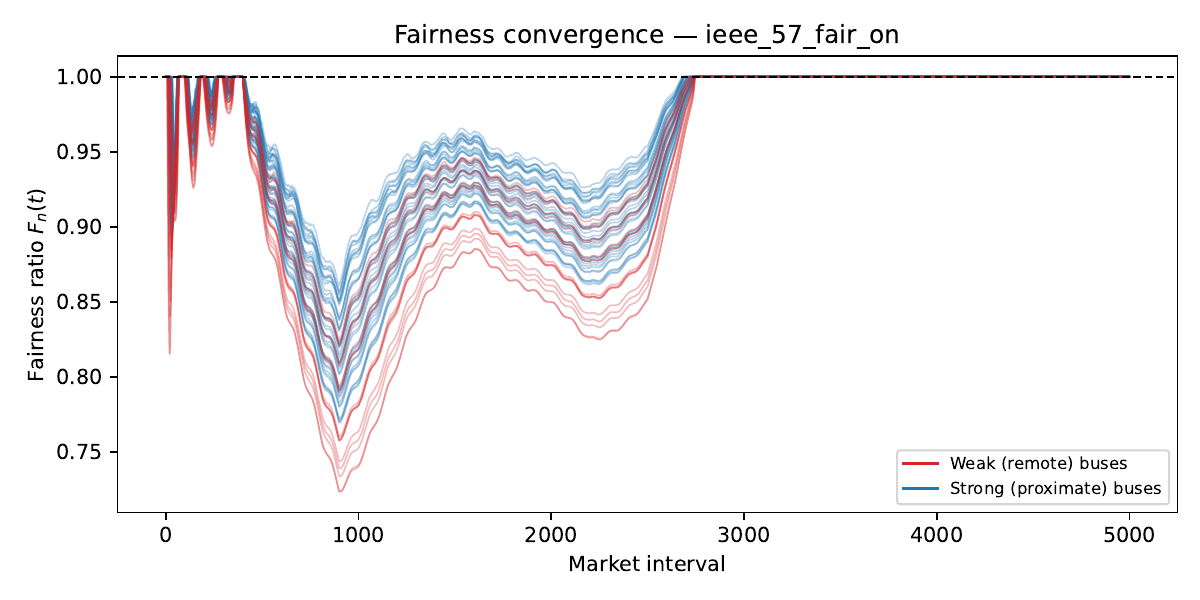}
\caption{\textbf{Per-bus fairness convergence, IEEE-57-bus (Fair
Play ON).} Red: weak (remote) buses; blue: strong (proximate)
buses. All 57 buses converge to $F^\star = 1$ despite the
persistent $1.8\times$ demand disadvantage on weak buses, directly
illustrating Theorem~\ref{thm:fairness_convergence}(a).}
\label{fig:ieee_57_fair_on}
\end{figure}

\subsection{Fair Play ON vs OFF: Weak-Bus Correction}

Figure~\ref{fig:ieee_vs_off_panel}
compares the maximum fairness error and mean weak-bus error under
Fair Play ON (solid blue) and Fair Play OFF (dashed red) for each
benchmark. The pattern is consistent across all three networks.
During scarcity windows (roughly intervals 400--900 and
1600--2200), Fair Play ON maintains a materially lower maximum
error than Fair Play OFF: without the shortage-memory priority
correction, weak buses systematically receive less than their
demanded share, and the imbalance persists throughout the scarcity
period. The separation is most pronounced on the IEEE-57 network,
where Fair Play OFF reaches a maximum error of $0.62$ during the
first scarcity window while Fair Play ON remains below $0.28$---a
54\% reduction that represents the starkest illustration of the
practical value of shortage-memory correction in a topologically
diverse network. With Fair Play ON, the shortage-memory state
$z_n(t)$ accumulates on under-served weak buses and increases
their selection probability, reducing the peak gap and accelerating
recovery once supply is restored. After the recovery windows, both
mechanisms converge to $F^\star$; the distinction is entirely in
the scarcity-period behaviour, where the shortage-memory mechanism
is most needed.

\begin{figure*}[t]
\centering
\includegraphics[width=\textwidth]{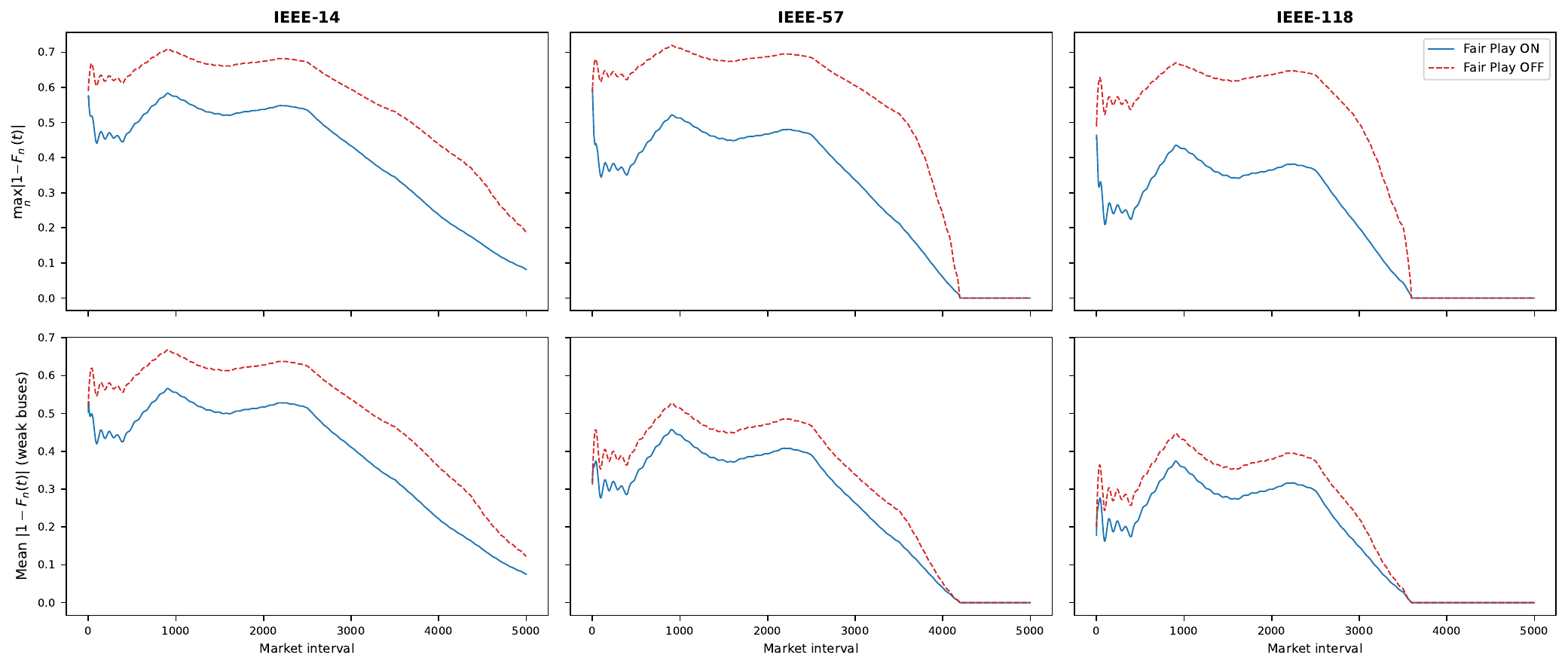}
\caption{\textbf{Fair Play ON vs OFF: weak-bus correction across all benchmarks.}
Top row: $\max_n |1 - F_n(t)|$; bottom row: mean weak-bus error $|1 - F_n(t)|$.
Columns left to right: IEEE-14, IEEE-57, IEEE-118. Solid blue: Fair Play ON; dashed
red: Fair Play OFF. During scarcity windows ($t \approx 400$--$900$ and
$1600$--$2200$), Fair Play ON maintains materially lower peak error on all three
networks. The separation is strongest on IEEE-57, where Fair Play OFF reaches $0.62$
while Fair Play ON remains below $0.28$ (54\% reduction), consistent with its highest
estimated drift constant $c_h = 0.18$ (Table~\ref{tab:ieee_summary}). After recovery
windows, both mechanisms converge to $F^\star$; the distinction is entirely in the
scarcity-period behaviour.}
\label{fig:ieee_vs_off_panel}
\end{figure*}

\subsection{DC Feasibility and Event-Trigger Rate}

Figure~\ref{fig:ieee_loading} confirms that maximum network loading
never exceeds the 70\% limit on any benchmark, validating the DC
feasibility enforcement of Algorithm~\ref{alg:fair_play}. Active
constraint binding occurs during the scarcity windows, when the
fairness mechanism is most needed; loading is conservative otherwise.
Figure~\ref{fig:ieee_trigger} shows the event-trigger update rate
across network sizes. The rate varies non-monotonically (73\%,
42\%, 68\% for IEEE-14, -57, -118 respectively), reflecting
network-specific state-change dynamics rather than system size. In
all cases the rate is below 100\%, confirming that the
event-triggered mechanism skips at least 25\% of clearing
intervals, consistent with Theorem~\ref{thm:event_trigger}.

\begin{figure}[t]
\centering
\includegraphics[width=\linewidth]{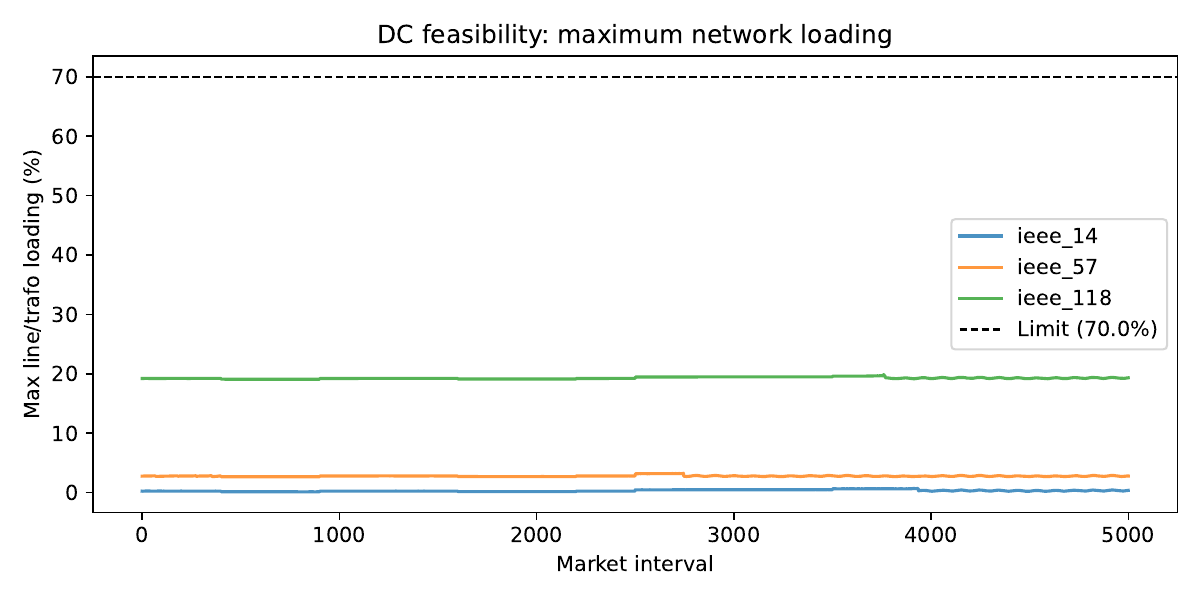}
\caption{\textbf{DC feasibility: maximum network loading (Fair Play
ON).} Maximum line/transformer loading on all three benchmarks
remains well below the 70\% feasibility limit throughout the
$T = 5000$ simulation.}
\label{fig:ieee_loading}
\end{figure}

\begin{figure}[t]
\centering
\includegraphics[width=\linewidth]{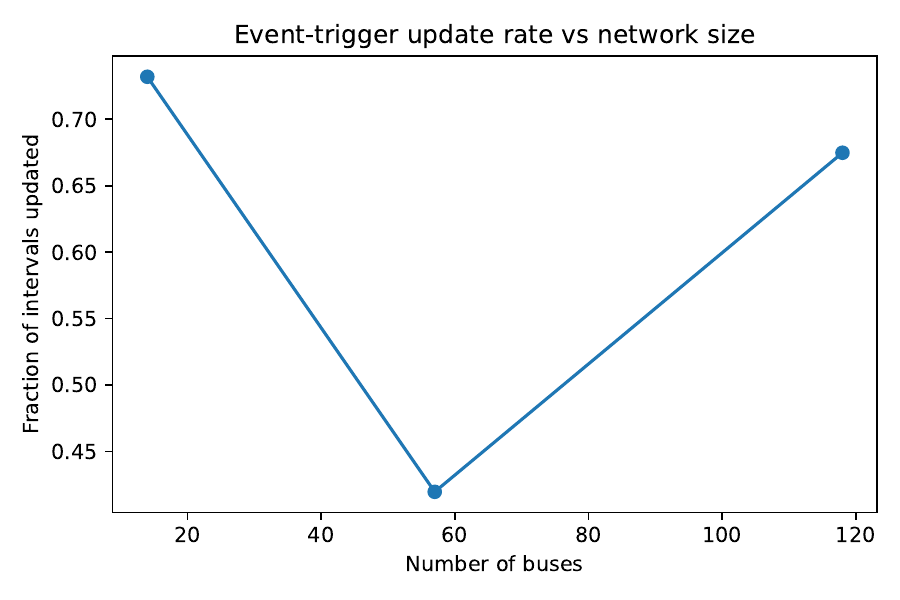}
\caption{\textbf{Event-trigger update rate vs network size.}
Fraction of intervals in which the clearing operator is
re-solved. The rate varies non-monotonically with network size,
reflecting network-specific state-change dynamics, and never
reaches 100\%, confirming the event-triggered reduction in
computation consistent with Theorem~\ref{thm:event_trigger}.}
\label{fig:ieee_trigger}
\end{figure}

\subsection{Summary}

Across IEEE-14, -57, and -118, the FP-AMM: (i) converges to
$F^\star = 1$ on all buses, including structurally disadvantaged
weak buses, validating Theorem~\ref{thm:fairness_convergence};
(ii) maintains a lower peak fairness error during scarcity windows
than the equal-weight baseline, confirming that the shortage-memory
correction is active and effective; and (iii) respects DC-OPF
feasibility constraints at every interval, with the event-triggered
clearing operator skipping at least 25\% of intervals.

% ======================================================================
\section{Discussion}
\label{sec:discussion}
% ======================================================================

\textbf{Architectural rationale.}
Fairness enters the FP-AMM exclusively through the stochastic priority
scores~\eqref{eq:priority_score}, not as an additive perturbation to
the gradient operator. This separation is deliberate: placing fairness
in the scores that govern \emph{which} requests are attempted, rather
than in the operator governing \emph{how} allocations are updated,
decouples the two mechanisms and enables Theorems~\ref{thm:operator_contraction}
and~\ref{thm:fairness_convergence} to be proved independently
(Remark~\ref{rem:separation}). The convergence guarantee also
requires genuine statefulness: standard LMP-based clearing is
memoryless---$F_n(t)$ never enters the clearing programme---so no
static clearing stack can guarantee convergence of cumulative delivery
ratios. The FP-AMM carries $z_t$ as explicit state, creating the
positive drift of Lemma~\ref{lem:drift_verification} that drives
Theorem~\ref{thm:fairness_convergence}.

\textbf{Design guidelines.}
The memory step $\beta \in (0.05, 0.2)$ gives 2--10 interval
correction horizons for 30-minute markets; $\alpha_f > 1$ amplifies
correction for severely under-served nodes. The pair $(\delta, K)$
tunes computation against fidelity via~\eqref{eq:ultimate_bound}:
the IEEE benchmarks use $\delta = 1.5/\sqrt{N}$, achieving 42--73\%
trigger rates with no loss in fairness convergence.

\textbf{Limitations.}
Theorem~\ref{thm:fairness_convergence} assumes persistent participation
(Assumption~\ref{ass:participation}); intermittent participants reset
their fairness ratio with potentially slower convergence. Extension
to non-stationary supply would require modifications to the martingale
argument. Extension to AC-OPF requires convex relaxation or local
Lipschitz arguments; Theorem~\ref{thm:lipschitz} applies locally.
Bounding $L_x$ analytically for large networks remains open.

% ======================================================================
\section{Conclusion}
\label{sec:conclusion}
% ======================================================================

This paper has developed and analysed the Fair Play Automatic Market
Maker (FP-AMM), a programmable electricity allocation mechanism in
which fairness, service differentiation, and network feasibility are
composable, formally analysable computational layers.

The main results are: (i) the shortage-memory state is invariant and
contractive (Theorem~\ref{thm:memory_bounded}); (ii) the intra-interval
clearing operator converges linearly to a unique fixed point with
explicit rate $q^k$ (Theorem~\ref{thm:operator_contraction}), and
the resulting fixed point is Lipschitz in the market state
(Theorem~\ref{thm:lipschitz}); (iii) the per-node delivery ratio
converges almost surely to $F^\star$ at rate $O(1/\sqrt{T})$, with
tier-level separation maintained and explicit finite-time constant
$C_\rho$ (Theorem~\ref{thm:fairness_convergence}); and (iv)
event-triggered execution yields practical ultimate boundedness of
the allocation tracking error (Theorem~\ref{thm:event_trigger}).

The stochastic two-stage clearing rule---tier sampling followed by
inverse-fairness weighting---provides a stochastic allocation
architecture for constrained cyber-physical networks with dynamic
fairness guarantees, replacing the classical reliance on static market
equilibrium with a unified closed-loop formulation in which physical
feasibility, allocation, and historical fairness evolve jointly as
system dynamics. The FP-AMM therefore demonstrates that dynamic
fairness in electricity allocation is fundamentally a
state-estimation and feedback-control problem rather than solely a
pricing problem.

\bibliographystyle{IEEEtran}

\end{document}